\documentclass[sigart,screen]{acmart}
\usepackage{amsthm}
\usepackage{xcolor}
\usepackage{hyperref}
\usepackage{qcircuit}
\usepackage{tikz}

\usepackage[T1]{fontenc}

\usepackage{complexity}

\allowdisplaybreaks 

\usepackage{tikz}
\usetikzlibrary{shapes,backgrounds}

\newtheorem{theorem}{Theorem}
\newtheorem{definition}[theorem]{Definition}

\newtheorem{lemma}[theorem]{Lemma}

\newtheorem{fact}[theorem]{Fact}

\newtheorem{property}{Property}

\usepackage{graphicx, subfigure}
\usepackage{enumitem}
\usepackage{xparse}
\usepackage{physics}
\usepackage{color}

\newcommand{\II}{\mathbb{I}}

\newcommand{\X}{\mathsf{x}}
\newcommand{\Z}{ \mathsf{z}}
\renewcommand{\H}{\mathbf{H}}

\newcommand{\Exp}{\mathop{\mathbf{E}\hspace{0.13em}}}

\newcommand{\eps}{\epsilon}

\newcommand{\half}{\frac{1}{2}}

\newcommand{\defeq}{\mathrel{\overset{\makebox[0pt]{\mbox{\normalfont\tiny\sffamily def}}}{=}}}

\renewcommand{\exp}{\mathsf{exp}}

\newcommand{\bits}{\{0,1\}}

\renewcommand{\eqref}[1]{\textrm{eq.~}(\ref{#1})}

\renewcommand{\exp}{\mathrm{exp}}

\newcommand{\cay}{\operatorname{Cay}}

\newenvironment{proof_of}[1]{\noindent {\bf Proof of #1:}
    \hspace*{1mm}}{\hspace*{\fill} $\Box$ }

\begin{document}

\title{NLTS Hamiltonians from Good Quantum Codes}

\author{Anurag Anshu}
\affiliation{%
  \institution{School of Engineering and Applied Sciences, Harvard University}
  \city{Cambridge, Massachusetts}
  \country{USA}}
\email{anuraganshu@fas.harvard.edu}

\author{Nikolas P. Breuckmann}
\affiliation{%
  \institution{School of Mathematics, University of Bristol}
  \city{Bristol}
  \country{United Kingdom}}
\email{niko.breuckmann@bristol.ac.uk}

\author{Chinmay Nirkhe}
\affiliation{%
  \institution{IBM Quantum, MIT-IBM Watson AI Lab}
  \city{Cambridge, Massachusetts}
  \country{USA}}
\email{nirkhe@ibm.com}

\begin{abstract}
The NLTS (No Low-Energy Trivial State) conjecture of Freedman and Hastings posits that there exist families of Hamiltonians with all low energy states of non-trivial complexity (with complexity measured by the quantum circuit depth preparing the state).
We prove this conjecture by showing that families of constant-rate and linear-distance qLDPC codes correspond to NLTS local Hamiltonians.
\end{abstract}



\maketitle

\footnotetext{This work has been published in the proceedings of STOC ’23, June 20–23, 2023, Orlando, FL, USA.}

\section{Introduction}

Ground- and low-energy states of local Hamiltonians are the central objects of study in condensed matter physics. For example, the $\QMA$-complete local Hamiltonian problem is the quantum analog of the $\NP$-complete constraint satisfaction problem (CSP) with ground-states (or low-energy states) of local Hamiltonians corresponding to solutions (or near-optimal solutions) of the problem~\cite{KSV02}. A sweeping insight into the computational properties of the low energy spectrum is embodied in the quantum PCP conjecture, which is arguably the most important open question in quantum complexity theory \cite{10.1145/2491533.2491549}. Just as the classical PCP theorem establishes that CSPs with constant fraction promise gaps remain $\NP$-complete, the quantum PCP conjecture asserts that local Hamiltonians with a constant fraction promise gap remain $\QMA$-complete. 
Despite numerous results providing evidence both for \cite{10.1145/1536414.1536472,FH14,NatarajanV18} and against \cite{10.5555/2011637.2011639,10.1145/2488608.2488719,10.1007/s11128-014-0877-9} the validity of the quantum PCP conjecture, the problem has remained open for nearly two decades. 

The difficulty of the quantum PCP conjecture has motivated a flurry of research beginning with Freedman and Hastings' \emph{No low-energy trivial states (NLTS) conjecture} \cite{FH14}. The NLTS conjecture posits that there exists a fixed constant $\eps > 0$ and a family of $n$ qubit local Hamiltonians such that every state of energy $\leq \eps n$ requires a quantum circuit of super-constant depth to generate. The NLTS conjecture is a necessary consequence of the quantum PCP conjecture, because $\QMA$-complete problems do not have $\NP$ solutions and a constant-depth quantum circuit generating a low-energy state would serve as a $\NP$ witness. Thus, this conjecture addresses the inapproximability of local Hamiltonians by classical means.

Previous progress \cite{8104078,nirkhe_et_al:LIPIcs:2018:9095,Eldar2019RobustQE,BravyiKKT19,anshu_et_al:LIPIcs.ITCS.2022.6,comb-nlts} provided solutions to weaker versions of the NLTS conjecture, but the complete conjecture 
had eluded the community. For a broader survey on the NLTS conjecture including the resolution proved in this work, we recommend the Ph.D. thesis of Nirkhe \cite{Nirkhe:EECS-2022-236}.

\begin{theorem}[No low-energy trivial states]
There exists a fixed constant $\eps > 0$ and an explicit family of $O(1)$-local frustration-free commuting Hamiltonians $\{\H^{(n)}\}_{n = 1}^\infty$ where $\H^{(n)} = \sum_{i = 1}^m h_i^{(n)}$ acts on $n$ particles and consists of $m = \Theta(n)$ local terms such that for any family of states $\{\psi_n\}$ satisfying $\tr(\H^{(n)} \psi) < \eps n$, the circuit complexity of the state $\psi_n$ is at least $\Omega(\log n)$.
\end{theorem}

The local Hamiltonians for which we can show such robust circuit-lower bounds correspond to constant-rate and linear-distance quantum LDPC error-correcting codes with an additional property related to the clustering of approximate code-words of the underlying classical codes. 
We show that the property holds for families of good qLDPC codes. 
While we show that the property is sufficient for NLTS, it is an interesting open question if the property is inherently satisfied by all constant-rate and linear-distance constructions.

\vspace{0.1in}

\noindent\textbf{Quantum code.} To formalize this property, recall a CSS code with parameters $[[n,k,d]]$. The code is constructed by taking two classical codes $C_{\X}$ and $C_{\Z}$ such that $C_{\Z} \supset C_{\X}^\perp$.
{The code $C_{\Z}$ is the kernel of a row- and column-sparse matrix $H_{\Z} \in \mathbb{F}_2^{m_\Z \times n}$; likewise, $C_{\X}$ and $H_{\X}\in \mathbb{F}_2^{m_\X \times n}$.}
The rank of $H_{\Z}$ will be denoted as~$r_{\Z}$ and likewise~$r_{\X}$ is the rank of $H_{\X}$.
Therefore, $n = k + r_{\X} + r_{\Z}$.
If the code is constant-rate and linear-distance, then $k, d, r_{\Z}, r_{\X} = \Omega(n)$. {For the codes considered in this work, we also have $m_{\Z}, m_{\X}=\Omega(n)$.}\\

\noindent For any subset $S \subset \bits^n$, define a distance measure $\abs{\cdot}_S$ as $\abs{y}_S = \min_{s \in S} \abs{y + s}$ where $\abs{\cdot}$ denoted Hamming weight. 
We define $G_{\Z}^\delta$ as the set of vectors which violate at most a $\delta$-fraction of checks from $C_{\Z}$, i.e. {$G_{\Z}^\delta = \{y : \abs{H_{\Z} y} \leq \delta m_{\Z}\}$}. 
We similarly define $G_{\X}^\delta$.

\begin{property}[Clustering of approximate code-words]\label{property}
We say that a $[[n,k,d]]$ CSS code defined by classical codes $(C_{\X}, C_{\Z})$ \emph{clusters approximate code-words} if there exist constants $c_1, c_2, \delta_0$ such that for sufficiently small $0 \leq \delta < \delta_0$ and every vector $y \in \bits^n$,
\begin{enumerate}
    \item If $y\in G_{\Z}^\delta$, then either $\abs{y}_{C_{\X}^\perp} \leq c_1 \delta n$ \ or else \ $\abs{y}_{C_{\X}^\perp} \geq c_2 n$.
    \item If $y\in G_{\X}^\delta$, then either $\abs{y}_{C_{\Z}^\perp} \leq c_1 \delta n$ \ or else \ $\abs{y}_{C_{\Z}^\perp} \geq c_2 n$.
\end{enumerate}
\end{property}

Note that this property holds for classical Tanner codes with spectral expansion (see \cite[Theorem 4.3]{comb-nlts}) and was used to prove the combinatorial NLTS conjecture.
In Section~\ref{sec:qLDPC} we discuss the connection between Property~\ref{property} and expansion properties of parity check matrices of codes.

\vspace{0.1in}

\noindent\textbf{Local Hamiltonian definition.} The aforementioned quantum codes lead to a natural commuting frustration-free local Hamiltonian.
{For every row $w_{\Z}$ of $H_{\Z}$ -- i.e. a stabilizer term $Z^{w_{\Z}}$ of the code, we associate a Hamiltonian term $\half(\II - Z^{w_{\Z}})$. We define $\H_{\Z}$ as the sum of all such terms for $H_{\Z}$.}
$\H_{\X}$ is defined analogously and the full Hamiltonian is $\H = \H_{\X} + \H_{\Z}$.
The number of local terms is $m_{\X} + m_{\Z} = \Theta(n)$ and $\H$ has zero ground energy.
We refer the reader to the preliminaries of \cite[Section 2]{anshu_et_al:LIPIcs.ITCS.2022.6} for more technical definitions and notation.

\vspace{0.1in}

\noindent\textbf{Open questions.} There are three questions that we leave unanswered.
\begin{itemize}
    \item Does Property \ref{property} ``morally'' hold for all constant-rate and linear-distance quantum codes?
    \item Grigoriev \cite{Grigoriev01} constructs instances of $3$-XOR problem that cannot be well approximated by $O(n)$ degree sum-of-squares algorithms. The proof relies on expansion (see \cite[Lecture 3-2]{baraksteu}) to construct equivalence classes, similar to our construction of equivalence classes in Section \ref{sec:NLTSproof}. A recent improvement to Grigoriev's theorem \cite{sum-of-squares-codes} uses the small-set boundary and co-boundary expansion similar to Property \ref{property}. It would be interesting to better understand this connection between quantum circuit lower bounds and sum-of-squares lower bounds.
    \item Overlap Gap Property \cite{Gamarnik21} holds for some constraint satisfaction problems (CSPs) whose near-optimal solutions are either close or far away in hamming distance. It is widely used in proving lower bounds on `stable' classical algorithms for such CSPs. Overlap gap property bears resemblance to Property \ref{property} (when the distance measure is hamming weight), which suggests possible connections between the two lower bound techniques.  
    \item Our construction does \emph{not} require quantum local testability. Property \ref{property} is sufficient for clustering of the classical distributions of low-energy states but it is weaker than local testability. 
    \cite{8104078} used local testability to argue clustering for their proof that local testability implies NLTS. What are the implications of codes with Property \ref{property} for the quantum PCP conjecture \cite{10.1145/2491533.2491549}?
    \item Can our proof techniques be generalized to prove non-trivial lower bounds for non-commuting Hamiltonians?
\end{itemize}

\section{Proof of the NLTS theorem}
\label{sec:NLTSproof}

The proof, that the local Hamiltonian corresponding to a constant-rate and linear-distance code satisfying Property~\ref{property} is NLTS, is divided into a few steps. We first show that the classical distributions generated by measuring any low-energy state in the standard or Hadamard bases are approximately supported on a particular structured subset of vectors. Then, we show that the subsets cluster into a collection of disjoint components which are far in Hamming distance from each other. Finally, we show that the distribution in one of the two bases cannot be too concentrated on any particular cluster. This shows that the distribution is \emph{well-spread} which can be used to prove a circuit depth lower bound.

\vspace{0.1in}

\noindent\textbf{The supports of the underlying classical distributions.} Consider a state $\psi$ on $n$ qubits such that $\tr(\H \psi) \leq \eps n$. Let $D_{\X}$ and $D_{\Z}$ be the distributions generated by measuring the $\psi$ in the (Hadamard) $X-$ and (standard) $Z-$ bases, respectively. We find that $D_{\Z}$ is largely supported on $G_{\Z}^{O(\eps)}$. Formally, this is because, by construction,
\begin{equation*}
    \eps n \geq \tr(\H \psi) \geq \tr(\H_{\Z} \psi) = \Exp_{y \sim D_{\Z}} \abs{H_{\Z} y}.
\end{equation*}
{Here, the last equality holds since for a Pauli operator $Z^a$, $$\bra{y}\frac{\II-Z^a}{2}\ket{y}=\frac{1-(-1)^{a.y}}{2}=a.y .$$} Let $q \defeq D_{\Z}(G_{\Z}^{\eps_1})$ be the probability mass assigned by $D_{\Z}$ to $G_{\Z}^{\eps_1}$. Then,
\begin{equation*}
    \Exp_{y \sim D_{\Z}} \abs{H_{\Z} y} \geq 0 \cdot q + (1-q)\cdot \eps_1 {m_{\Z}} = (1-q) \eps_1 {m_{\Z}}.
\end{equation*}
Therefore, $D_{\Z}(G_{\Z}^{\eps_1}) \geq 1- \eps n / (\eps_1 {m_{\Z}})$. 
A similar argument shows that $D_{\X}(G_{\X}^{\eps_1}) \geq 1- \eps n / (\eps_1 {m_{\X}})$. 
With the choice $\eps_1 = \frac{200n}{ \min\{{m_{\X},m_{\Z}}\}} \cdot \eps$, we find 
\begin{equation*}
\label{eq:largemass}
D_{\Z}(G_{\Z}^{\eps_1}), D_{\X}(G_{\X}^{\eps_1}) \geq \frac{199}{200}
\end{equation*}
for both the bases.  

\vspace{0.1in}

\noindent\textbf{The supports are well clustered.} Given that $D_{\Z}$ is well supported on $G_{\Z}^{\eps_1}$, it is helpful to understand the structure of $G_{\Z}^{\eps_1}$. For $x, y \in G_{\Z}^{ \eps_1}$, notice that $x \oplus y \in G_{\Z}^{ 2\eps_1}$ since $x \oplus y$ satisfies every check that both $x$ and $y$ satisfy. By Property \ref{property} (and assuming $2\eps_1\leq \delta_0$), then either
\begin{equation*}
    \abs{x \oplus y}_{C_{\X}^\perp} \leq 2 c_1 \eps_1 n \quad \text{or else} \quad \abs{x \oplus y}_{C_{\X}^\perp} \geq c_2 n.
\end{equation*}
Define a relation `$\sim$' such that for $x,y \in G_{\Z}^{\eps_1}$, $x \sim y$ iff $\abs{x \oplus y}_{C_{\X}^\perp} \leq 2 c_1 \eps_1 n$.
To prove that the relation is transitive and therefore an equivalence relation, notice that if $x \sim y$ and $y \sim z$, then
\begin{equation*}
\abs{x \oplus z}_{C_{\X}^\perp} \leq \abs{x \oplus y}_{C_{\X}^\perp} + \abs{y \oplus z}_{C_{\X}^\perp} \leq 4c_1 \eps_1 n.
\end{equation*}
However, $x \oplus z \in G_{\Z}^{2\eps_1}$ and for sufficiently small $\eps_1$ such that $4c_1 \eps_1 < c_2$, Property \ref{property} implies that $\abs{x \oplus z}_{C_{\X}^\perp} \leq 2c_1\eps_1n$. Thus, $x \sim z$ and hence $\sim$ forms an equivalence relation. We can now divide the set $G_{\Z}^{\eps_1}$ into clusters $B_{\Z}^{1}, B_{\Z}^{2}, \ldots,$ according to the equivalence relation $\sim$. Furthermore, the distance between any two clusters is $\geq c_2 n$, since for $x$ in one cluster and $x'$ in another cluster, we have $\abs{x\oplus x'} \geq \abs{x\oplus x'}_{C_{\X}^\perp}\geq c_2n$. Lastly, the same argument holds for~$G_{\X}^{\eps_1}$. 

\vspace{0.1in}

\noindent\textbf{The distributions are not concentrated on any one cluster.}
To apply known circuit-depth lower bounding techniques to $D_{\Z}$, it suffices to show that $D_{\Z}$ is not concentrated on any one cluster~$B_{\Z}^i$.
However, it is not immediate how to show this property for $D_{\Z}$. Instead, what we can show is that is impossible for both $D_{\Z}$ to be concentrated on any one cluster $B_{\Z}^i$ and $D_{\X}$ to be concentrated on any one cluster $B_{\X}^j$. 

\begin{lemma}
\label{lem:noconc}
For $\eps_1$ such that $2 c_1 \eps_1 \leq \qty(\frac{k-1}{4n})^2$,
either $\forall \ i$, $D_{\Z}(B_{\Z}^i) < 99/100$ or else $\forall \ j$, $D_{\X}(B_{\X}^j) < 99/100$.
\end{lemma}

\begin{proof}
Assume there exists some $i$ such that $D_{\Z}(B_{\Z}^i) \geq 99/100$. We will employ the following fact that captures the well-known uncertainty of measurements in the standard and Hadamard bases; a proof is provided in the appendix.
\begin{fact}
Consider a state $\psi$ and corresponding measurement distributions $D_{\X}$ and $D_{\Z}$. For all subsets $S,T \subset \bits^n$ it holds that  $D_{\X}(T) \leq 2\sqrt{1 - D_{\Z}(S)} + |S|\cdot |T|/2^n$.
\label{fact:fourier-transform}
\end{fact}
For any $j$, we employ this fact with $S = B_{\Z}^i$ and $T = B_{\X}^j$. To bound $|B_{\Z}^i|$, fix any string $z\in B_{\Z}^i$. Any other string $z'\in B_{\Z}^i$ has the property that its Hamming distance from $z\oplus w$ (for some $w\in C_{\X}^{\perp}$) is at most $2 c_1 \eps_1 n$. Since $\abs{C_{\X}^\perp} = 2^{\dim C_{\X}^\perp} = 2^{n - \dim C_{\X}} = 2^{r_{\X}}$,
the size of the cluster $B_{\Z}^i$ is at most
\begin{equation*}
    2^{r_{\X}} \cdot {n \choose 2 c_1 \eps_1 n} \leq 2^{r_{\X}} \cdot 2^{2 \sqrt{2 c_1 \eps_1}n}.
\end{equation*}
A similar bound can be calculated of $\abs{B_{\X}^j}\leq 2^{r_{\Z}} \cdot 2^{2 \sqrt{2 c_1 \eps_1}n}$. Then applying Fact \ref{fact:fourier-transform} with the bound on $\eps_1$ as stated in the Lemma,
\begin{equation*}
    \forall \ j, \quad D_{\X}\qty(B_{\X}^j) \leq \frac{1}{5} +  2^{r_{\X} + r_{\Z} - n} \cdot 2^{4 \sqrt{2 c_1 \eps_1} n} = \frac{1}{5} + 2^{-k + 4\sqrt{2 c_1 \eps_1} n} < \frac{99}{100}.
\end{equation*}
\end{proof}

\vspace{0.1in}

\noindent\textbf{A lower bound using the \emph{well-spread} nature of the distribution.}
Assume, without loss of generality, from Lemma~\ref{lem:noconc} that $D_{\Z}$ is not too concentrated on any cluster $B_{\Z}^i$. Recall that $D_{\Z}(\bigcup_i B_{\Z}^i) \geq 199/200$. Therefore, there exist disjoint sets $M$ and $M'$ such that\footnote{Consider building the set $M$ greedily by adding terms until the mass exceeds $1/400$. Upon adding the final term to overcome the threshold, the total mass is at most $397/400$ since no term is larger than $99/100$. Therefore, the remainder of terms not included in $M$ must have a mass of at least $199/200 - 397/400 = 1/400$.}
\begin{equation*}
    D_{\Z}\qty(\bigcup_{i \in M} B_{\Z}^i) \geq \frac{1}{400} \quad \text{and} \quad D_{\Z}\qty(\bigcup_{i \in M'} B_{\Z}^i) \geq \frac{1}{400} .
\end{equation*}
Furthermore, recall that since the distance between any two clusters is at least $c_2 n$, the same distance lower bound holds for the union of clusters over $M$ and $M'$ as well. This proves that the distribution $D_{\Z}$ is \emph{well-spread} which implies a circuit lower bound due to the following known fact\footnote{Versions of this lower-bound for well-spread distributions can be found in \cite{comb-nlts}, Theorem 4.6], \cite[Corollary 43]{8104078}, and \cite[Lemma 13]{ anshu_et_al:LIPIcs.ITCS.2022.6}.} (see Appendix for proof):
\begin{fact}
Let $D$ be a probability distribution on $n$ bits generated by measuring the output of a quantum circuit in the standard basis. If two sets $S_1, S_2 \subset \bits^{n}$ satisfy $D(S_1), D(S_2) \geq \mu$, then the depth of the circuit is at least
\begin{equation*}
  \frac{1}{3} \log \qty( \frac{ \mathrm{dist}(S_1, S_2)^2}{ 400 n \cdot \log\frac{1}{\mu}}).
\end{equation*}
\label{fact:well-spread-to-lb}
\end{fact}
An immediate application of this fact gives a circuit-depth lower bound of $\Omega(\log n)$ for $D_{\Z}$ since $\mathrm{dist}(S_1, S_2) \geq c_2 n$ and $\mu = \frac{1}{400}$. Since the circuit depth of $D_{\Z}$ is at most one more than the circuit depth of~$\psi$, the lower bound is proven.

\begin{theorem}[Formal statement of the NLTS theorem]
Consider a $[[n,k,d]]$ CSS code satisfying Property~\ref{property} with parameters $\delta_0, c_1, c_2$ as stated. Let $\H$ be the corresponding local Hamiltonian. Then for
\begin{equation*} 
    \eps < \frac{1}{400c_1} \qty(\frac{\min\{{m_{\X},m_{\Z}}\}}{n}) \cdot \min \left\{ \qty(\frac{k-1}{4n})^2, \delta_0, \frac{c_2}{2} \right\},
\end{equation*}
and every state $\psi$ such that $\tr(\H \psi) \leq \eps n$, the circuit depth of $\psi$ is at least $\Omega(\log n)$. For constant-rate and linear-distance codes satisfying\footnote{While the distance parameter $d$ does not appear in the bound on $\eps$, Property \ref{property} for $\delta = 0$ implies linear distance.} Property \ref{property}, the bound on $\eps$ is a constant.
\end{theorem}

\section{Proof that Property~\ref{property} holds for families of good qLDPC codes}\label{sec:qLDPC}
In this section we show that families of good qLDPC codes have Property~\ref{property}, which concludes our proof.
First, we show that a certain expansion property of binary matrices, which generalizes small-set ((co-)boundary) expansion, gives rise to a clustering property.
This clustering property readily implies Property~\ref{property}.

\vspace{0.1in}

\noindent\textbf{Expansion implies clustering.}
In this section we are going to show that an expansion property of binary matrices implies a clustering property.
It follows that Property~\ref{property} holds for families of good qLDPC by as their check matrices are expanding.
For the remainder of this section let $V \subseteq \mathbb{F}_2^n$.
\begin{definition}\label{def:expanding_matrix}
    We call a matrix $A \in \mathbb{F}_2^{m\times n}$ \emph{$(\alpha,\beta)$-expanding on~$V$} if for all $x \in V$ with $|x| \leq \alpha n$ it holds that $|Ax|\geq \beta |x|$.
\end{definition}
Note that Definition~\ref{def:expanding_matrix} gives small-set \mbox{(co-)}minimal \mbox{(co-)}boundary expansion~\cite{linial2006homological} when~$A$ is a (co-)boundary operator and~$V$ is chosen to only contain (co-)chains that are minimal with respect to (co-)boundaries.
\begin{definition}\label{def:Vclustering}
    We call a matrix $A \in \mathbb{F}_2^{m\times n}$ \emph{$(\delta,\mu,\nu)$-clustering on~$V$} if for all $x \in V$ with $|Ax| \leq \delta n$ it holds that either $|x| \leq \mu n$ or $|x| \geq \nu n$.
\end{definition}
The following lemma shows that expansion implies clustering.
\begin{lemma}\label{lem:ExpandingToClustering}
    A matrix $A \in \mathbb{F}_2^{m\times n}$ that is $\left(\nu,\frac{\delta}{\mu}\right)$-expanding on~$V$ is $(\delta,\mu,\nu)$-clustering on $V$.
\end{lemma}
\begin{proof}
    Assume $A \in \mathbb{F}_2^{m\times n}$ is $\left(\nu,\frac{\delta}{\mu}\right)$-expanding on $V$.
    Take $x \in V$ with $|Ax| \leq \delta n$.
    If $|x| \leq \mu n$ then we are done.

    Assume $|x| > \mu n$.
    It follows that
    \begin{align*}
        |Ax| \leq \delta n < \frac{\delta}{\mu} |x| .
    \end{align*}
    Since $A$ is $\left(\nu,\frac{\delta}{\mu}\right)$-expanding on $V$ it follows that $|x| > \nu n$.
\end{proof}
In order for the statement of Lemma~\ref{lem:ExpandingToClustering} to be non-trivial we require $\mu < \nu$, which is always possible to achieve for any $(\alpha,\beta)$-expanding matrix with $\alpha,\beta>0$ by choosing $0<\delta < \alpha \beta$.

We note that Lemma~\ref{lem:ExpandingToClustering} immediately implies \cite[Theorem 4.3]{comb-nlts} by setting $V=\mathbb{F}_2^n$ and observing that the parity check matrix of expander codes are $(\alpha,\beta)$-expanding.
Similarly, for any CSS quantum code $(C_{\X}, C_{\Z})$ with $m_\X,m_\Z = \Omega(n)$ we have that Property~\ref{property} follows from Lemma~\ref{lem:ExpandingToClustering} by setting 
$$V_{\Z}^\delta = \lbrace y\in G_{\Z}^\delta \mid \forall s\in C_{\X}^\perp : |y| \leq |y + s| \rbrace$$
i.e.\ considering only approximate code words that are minimal with respect to adding stabilizers, and showing that $H_{\Z}$ is expanding on~$V_{\Z}^\delta$ for sufficiently small $\delta$ and mutatis mutandis for the $X$-side, setting
$$V_{\X}^\delta = \lbrace y\in G_{\X}^\delta \mid \forall s\in C_{\Z}^\perp : |y| \leq |y + s| \rbrace .$$

\vspace{0.1in}

\noindent\textbf{Balanced products and good qLDPC codes.}
Quantum codes can be constructed from classical codes using product constructions.
One such product is known as hypergraph product~\cite{5205648}, which is the tensor product of chain complexes when described in terms of homological algebra, see \cite{audoux2019tensor} and~\cite[Section~IV]{breuckmann2021quantum}.
However, quantum codes obtained this way can only yield distances $d \leq O(\sqrt{n})$, as code words of the classical input codes lift to logical operators in the product and the number of physical qubits of the product code grows as the product of the lengths of the input codes.

This limitation can be overcome by utilizing a more general product construction called balanced product \cite{Breuckmann2020BalancedPQ,breuckmann2021quantum}.
The idea is that, assuming a group $G$ acts on the input codes, we can take the tensor product and factor out the action of $G$, see \cite[Section~IV-E]{Breuckmann2020BalancedPQ} and \cite[Appendix~A]{breuckmann2021quantum} for details.

Suitable input codes can be obtained as follows:
for a group~$G$, consider a right Cayley graph $\cay^r (G,A)$ and a left Cayley graph $\cay^\ell (G,B)$ for two generating sets $A,B\subset G$, which are assumed to be symmetric, i.e.\ $A = A^{-1}$ and $B = B^{-1}$ and of the same cardinality $\Delta = |A| = |B|$.
Further, we define the double-covers of $\cay^r (G,A)$ and $\cay^\ell (G,B)$ that we will denote $X_2^r = \cay^r_2 (G,A)$ and $X_2^\ell = \cay^\ell_2 (G,B)$.
Assuming that $X^r_2$ and $X^\ell_2$ are sufficiently good spectral expanders, we can define Sipser-Spielman expander codes~\cite{sipser1996expander} from both of these graphs using local codes~$K$ and~$L$.
We denote these codes $C_\bullet(X_2^r,L)$ and $C_\bullet(X_2^\ell,K)$.\footnote{The notation is taken from homological algebra, as graph codes combined with local codes can be understood in terms of homology with local coefficients, as pointed out by Meshulam~\cite{meshulam2018graph}.}
We will assume that the Cayley graphs are Ramanujan graphs, so that for the spectrum of the adjacency matrix we have $\lambda = \max \{ |\lambda_2|, |\lambda_n| \} \leq 2 \sqrt{\Delta - 1}$, such as those constructed in \cite{lubotzky1988ramanujan,margulis1988explicit}.

In \cite{asymptotically-good-quantum-codes} Panteleev--Kalachev were the first to prove a linear distance bound for the balanced product between the code~$C_\bullet(X_2^r,K)$ and the dual code~$C_\bullet(X_2^r,L)^*$, denoted as $C_\bullet(X_2^r,K) \otimes_G C_\bullet(X_2^\ell,L)^*$, assuming a suitable choice of~$K$ and~$L$, see conjecture in \cite[Section~IV]{Breuckmann2020BalancedPQ}.
Similarly, Dinur~et~al.~\cite{linear-time-decoder-quantum-codes} proved a linear distance bound for the balanced product code $C_\bullet(X_2^r,K) \otimes_G C_\bullet(X_2^\ell,L) = C_\bullet(X_2^r\times_G X_2^\ell, K\otimes L)$.
The right-hand side refers to the fact that we can associate a local tensor code $K\otimes L$ directly to the balanced product complex $X_2^r\times_G X_2^\ell$ \cite[Section~IV-B]{Breuckmann2020BalancedPQ}.
In~\cite{quantum-tanner-codes} Leverrier--Z\'emor consider a construction that they call `quantum Tanner code', where they assign local tensor codes in a different way to the balanced product complex $X_2^r\times_G X_2^\ell$, leading to a more symmetric quantum code.

\vspace{0.1in}

\noindent\textbf{Expansion properties.}
For the remainder of this section let~$d_{\operatorname{loc}}$ be the smaller of the minimum distances of~$K$ and~$L$, as well as~$m_K$ and~$m_L$ their respective number of checks.

All distance proofs utilize a property of the local codes called robustness, which roughly speaking makes it possible to bound the weight of elements of the dual tensor code via elements of its factors.
The local codes $K$ and $L$ are called $\rho$-robust if for all $c\in K\otimes \mathbb{F}_2^{\Delta} + \mathbb{F}_2^{\Delta}\otimes L$ there exist $k\in K\otimes \mathbb{F}_2^{\Delta}$ and $l \in \mathbb{F}_2^{\Delta}\otimes L$ such that $c = k+l$ and $|c| \geq \rho (|k|_{\operatorname{row}} + |l|_{\operatorname{col}})$, where we interpret $c,k,l$ as matrices and $|\cdot|_{\operatorname{row}}$ and $|\cdot|_{\operatorname{col}}$ is the number of non-zero rows and columns, respectively.
It is shown in \cite[Theorem~1.2]{linear-time-decoder-quantum-codes} that randomized constructions give~$K$ and~$L$ such that they and their duals have linear distance and robustness.

The following theorem immediately follows from Theorem~3.8 and Corollary~3.10 in Ref.~\cite{linear-time-decoder-quantum-codes}.
\begin{theorem}\label{thm:CCexpansion}
    The parity-check matrices of the balanced product code $$C_\bullet(X_2^r,L)\otimes_G C_\bullet(X_2^\ell,K) = C_\bullet(X_2^r\times_G X_2^\ell, K\otimes L)$$ have the following expansion properties:
    \begin{itemize}
        \item $H_{\Z}$ is $\left( \alpha, \beta \right)$-expanding on $V_{\Z}^\delta$, with
        $$
            \alpha = \frac{\eta}{4 (m_K + m_L) \Delta}, \quad
            \beta = \frac{1}{\frac{2}{\eta^\perp} \Delta^3 + \Delta^2 + \Delta}
        $$
        and $0 < \delta < \alpha \beta$
        \item $H_{\X}$ is $\left( \alpha', \beta' \right)$-expanding on $V_{\X}^{\delta'}$, with
        $$
            \alpha' = \frac{\eta}{4 (m_K + m_L) \Delta}, \quad
            \beta' = \frac{\eta}{2}
        $$
        and $0 < \delta' < \alpha' \beta'$
    \end{itemize}
    where $$\eta = \frac{2 \rho d_{\operatorname{loc}} - 2 \lambda \rho -16 \lambda \Delta}{(\rho + 4) \Delta}$$
    and $\eta^\perp$ is defined similarly for the dualized local codes.
\end{theorem}
As discussed previously, Property~\ref{property} follows for the balanced product codes $C_\bullet(X_2^r,L)\otimes_G C_\bullet(X_2^\ell,K)$ by combining Theorem~\ref{thm:CCexpansion} and Lemma~\ref{lem:ExpandingToClustering}.

Similarly, for the construction considered by Leverrier-Z\'emor, it was shown by Hopkins--Lin that the parity check matrices are expanding, see \cite[Theorem 8.2]{hopkins2022explicit}.
\begin{theorem}
    Let $\epsilon \in (0,1/2)$
    then for $\Delta$ large enough it holds with probability approaching 1 that the parity-check matrices~$H_{\X}$ and~$H_{\Z}$ are both $\left(\alpha, \beta \right)$-expanding on $V_{\X}^\delta$ and $V_{\Z}^\delta$, respectively, where
    $$\alpha = \frac{{d_{\operatorname{loc}}}}{6\Delta^{5/2+\epsilon}} \text{ and } \beta = \frac{56}{\Delta^{3-2\epsilon}}$$
    and $0 < \delta < \alpha \beta$.
\end{theorem}

We state without proof that making sufficient adjustments to the distance proof of Panteleev--Kalachev~\cite{asymptotically-good-quantum-codes}, it can be shown that the parity check matrices of the code $C_\bullet(X_2^r,L)\otimes_G C_\bullet(X_2^\ell,K)^*$ are $(\alpha,\beta)$-expanding on suitably chosen subsets $V_\X$ and $V_\Z$ with $\alpha,\beta > 0$ not depending on $n$.
Hence, Property~\ref{property} also holds for this code family.

\section*{Acknowledgements}
We thank Nolan Coble, Matt Coudron, Jens Eberhardt, Lior Eldar, David Gamarnik, Sevag Gharibian, Yeongwoo Huang, Jon Nelson, Madhu Sudan and Umesh Vazirani for helpful discussions. The majority of this work was conducted while NPB was affiliated with University College London and CN was affiliated with the University of California, Berkeley. AA acknowledges support through the NSF CAREER Award No. 2238836 and NSF award QCIS-FF: Quantum Computing \& Information Science Faculty Fellow at Harvard University (NSF 2013303). NPB acknowledges support through the EPSRC Prosperity Partnership in Quantum Software for Simulation and Modelling (EP/S005021/1). CN was supported by NSF Quantum Leap Challenges Institute Grant number OMA2016245. Part of this work was completed while AA and CN were participants in the Simons Institute for the Theory of Computing \emph{The Quantum Wave in Computing: Extended Reunion}. 

\appendix

\section{Omitted Proofs}

\begin{proof_of}{Fact \ref{fact:fourier-transform}}
Consider a purification of the state $\psi$ as $\ket{\psi}$ on a potentially larger Hilbert space. Write $\ket{\psi}$ as $\sum_{z \in \bits^n} \ket{\psi_z}\otimes \ket{z}$ where the second register is the original $n$ qubit code-space. Define $C \defeq \sum_{z\in S} \|\ket{\psi_z}\|^2$ and
\begin{equation*}
    \ket{\psi'}= \frac{1}{\sqrt{C}}\sum_{z\in S} \ket{\psi_z}\otimes \ket{z}\defeq\sum_{z\in S} \ket{\psi'_z}\otimes \ket{z}.
\end{equation*}
Since $C=D_{\Z}(S) \defeq 1-\eta$, by the gentle measurement lemma \cite{Winter99} we have $\frac{1}{2}\|\ketbra{\psi}-\ketbra{\psi'}\|_1\leq 2\sqrt{\eta}$. Measuring $\ket{\psi'}$ in the computational basis, we obtain a string $z \in S$ with probability 
$\|\ket{\psi'_z}\|^2$. Measuring $\ket{\psi'}$ in the Hadamard basis, we obtain a string $x$ with probability 
\begin{equation*}p(x)\defeq \frac{1}{2^n}\left\|\sum_z (-1)^{\langle x, z \rangle}\ket{\psi'_z}\right\|^2 = \frac{1}{2^n}\left(\sum_{z,w} (-1)^{\langle x, z\oplus w \rangle }\braket{\psi'_w}{\psi'_z}\right).\end{equation*}
Using Cauchy-Schwarz,
\begin{align*}
   p(x)&= \frac{1}{2^n}\left(\sum_{z,w\in S} (-1)^{\langle x, z\oplus w \rangle}\braket{\psi'_w}{\psi'_z}\right)\\
   &\leq \frac{1}{2^n}\sqrt{|S|^2\sum_{z,w\in S}|\braket{\psi'_w}{\psi'_z}|^2}\\
   &\leq \frac{|S|}{2^n}\sqrt{\sum_{z,w\in S}\|\ket{\psi'_w}\|^2\|\ket{\psi'_z}\|^2} = \frac{|S|}{2^n}
\end{align*}
Thus, we find
\begin{equation*}
    \sum_{x\in T}p(x) \leq \frac{|T||S|}{2^n}.
\end{equation*}
Since $\frac{1}{2}\|\ketbra{\psi}-\ketbra{\psi'}\|_1\leq 2\sqrt{\eta}$, we conclude that
$D_{\X}(T) \leq 2\sqrt{\eta} +\frac{|S|\cdot |T|}{2^n}$.
\end{proof_of}

\vspace{0.1in}

\begin{proof_of}{Fact \ref{fact:well-spread-to-lb}}
Let $\ket{\rho}=U\ket{0}^{\otimes m}$ on $m\geq n$ qubits, where $U$ is a depth $t$ quantum circuit such that when $\ket{\rho}$ is measured in the standard basis, the resulting distribution is $p$. 
Note that $m\leq 2^t n$ without loss of generality (see \cite[Section 2.3]{anshu_et_al:LIPIcs.ITCS.2022.6} for a justification based on the light cone argument).
The Hamiltonian 
\begin{equation*}G = \Exp_{i = 1}^m U \ketbra{1}_i U^\dagger\end{equation*} has $\ket \rho$ as its unique ground-state, is commuting, has locality $2^t$, and has eigenvalues $0, 1/m, 2/m, \ldots 1$. There exists a polynomial $P$ of degree $f$, built from Chebyshev polynomials, such that $P(0) = 1$ and
\begin{equation*}
    \abs{P(i/m)} \leq \exp\qty(- \frac{f^2}{100m}) \leq  \exp\qty(- \frac{f^2}{100 \cdot 2^t n}) \ \forall \ i = 1, \ldots , m.
\end{equation*}
See \cite[Theorem 3.1]{AAG21} (or \cite{KahnLS96, BuhrmanCWZ99}) for details on the construction of~$P$. Applying the polynomial $P$ to the Hamiltonian $G$ results in an \emph{approximate ground-state projector}, $P(G)$, such that
\begin{equation*}\|\ketbra{\rho}-P(G)\|_\infty \leq  \exp\qty(- \frac{f^2}{100 \cdot 2^t n})\end{equation*}
Furthermore, $P(G)$ is a $f \cdot 2^t$ local operator. Setting $u \defeq \mathrm{dist}(S_1, S_2)$ and choosing $f \defeq \frac{u}{2^{t+1}}$, we obtain
\begin{equation*}\|\ketbra{\rho}-P(G)\|_\infty\leq \exp\qty({-\frac{u^2}{400 \cdot 2^{3t}  n}}).\end{equation*}
Let $\Pi_{S_1}, \Pi_{S_2}$ be projections onto the strings in sets $S_1, S_2$ respectively. Note that  $\Pi_{S_1} P(G) \Pi_{S_2}=0$, which implies 
\begin{equation*}\|\Pi_{S_1}\ketbra{\rho}\Pi_{S_2}\|_\infty \leq \exp\qty({-\frac{u^2}{400\cdot 2^{3t}\cdot n}}).\end{equation*}
However 
\begin{align*}
&\|\Pi_{S_1}\ketbra{\rho}\Pi_{S_2}\|_\infty=\sqrt{\bra{\rho}\Pi_{S_1}\ket{\rho}\cdot \bra{\rho}\Pi_{S_2}\ket{\rho}} = \sqrt{p(S_1) p(S_2)} \geq \mu.    
\end{align*}
Thus, 
$ 2^{3t} \geq \frac{u^2}{400\cdot \log\frac{1}{\mu} \cdot n}$, which rearranges into the desired statement.
\end{proof_of}

\bibliographystyle{ACM-Reference-Format}
\bibliography{references}

\end{document}